\documentclass[twoside,leqno,twocolumn]{article}
\usepackage{ltexpprt}
\usepackage{float}
\usepackage{amssymb}
\usepackage{amsmath}
\usepackage[sort,nocompress]{cite}
\usepackage{caption}
\usepackage{mathtools}

\newcommand{\qed}{\hfill$\Box$}
\newcommand{\bigO}{\mathcal{O}}
\newcommand{\llen}{\ell}

\newcommand{\rk}{\phi_q}
\newcommand{\seed}{\bar s}

\begin{document}

\title{\Large In-Place Sparse Suffix Sorting\thanks{
		Part of this work was done while the author was a PhD student at the University of Udine, Italy. 
		Work supported by the Danish Research Council (DFF-4005-00267).}}
\author{Nicola Prezza\thanks{DTU Compute, Technical University of Denmark}}
\date{}

\maketitle


\fancyfoot[R]{\footnotesize{\textbf{Copyright \textcopyright\ 2018 by SIAM\\
Unauthorized reproduction of this article is prohibited}}}





\begin{abstract} \small\baselineskip=9pt 
Suffix arrays encode the lexicographical order of all  suffixes of a text and  are often combined with the Longest Common Prefix array (LCP) to simulate navigational queries on the suffix tree in reduced space. In space-critical applications such as sparse and compressed text indexing, only information regarding the lexicographical order of a size-$b$ subset of all $n$ text suffixes is often needed. Such information can be stored space-efficiently (in $b$ words) in the sparse suffix array (SSA). The SSA and its relative sparse LCP array (SLCP) can be used as a space-efficient substitute of the sparse suffix tree. Very recently,  Gawrychowski and Kociumaka~\cite{gawrychowski2017sparse} showed that the sparse suffix tree (and therefore SSA and SLCP) can be built in asymptotically optimal $\bigO(b)$ space with a Monte Carlo algorithm running in $\bigO(n)$ time. The main reason for using the SSA and SLCP arrays in place of the sparse suffix tree is, however, their reduced space of $b$ words each. This leads naturally to the quest for in-place algorithms building these arrays. Franceschini and Muthukrishnan~\cite{franceschini2007place} showed that the full suffix array can be built in-place and in optimal running time. On the other hand, finding sub-quadratic in-place algorithms for building the SSA and SLCP for \emph{general} subsets of suffixes has been an elusive task for decades. In this paper, we give the first solution to this problem. We provide the first in-place algorithm building the full LCP array in $\bigO(n\log n)$ expected time and the first Monte Carlo in-place algorithms building the SSA and SLCP in $\bigO(n + b\log^2 n)$ expected time. 
We moreover describe the first in-place solution for the suffix selection problem: to compute the $i$-th smallest text suffix.
In order to achieve these results, we show that we can quickly overwrite the text with a reversible and implicit data structure supporting Longest Common Extension queries in polylogarithmic time and text extraction in optimal time: this structure is strictly more powerful than a plain text representation and is of independent interest.
\end{abstract}

\section{Introduction}

The \emph{suffix sorting} problem --- to compute the lexicographic order of all suffixes of a text --- has been the subject of study of dozens of research articles since the introduction of suffix arrays in~\cite{manber1993suffix,baeza1992new,gonnet1992new}, and is a fundamental step in most of the indexing and compression algorithms developed to date.
The algorithm originally introduced by Manber and Myers in~\cite{manber1993suffix} sorts all suffixes of a text of length $n$ in $\bigO(n)$ words of space and $\bigO(n\log n)$ time. Even faster construction times can be achieved by constructing the suffix tree in linear time\cite{ukkonen1995line} and traversing it. 
Reducing the constants hidden in the asymptotic linear running time and linear space occupancy of suffix sorting algorithms has been one of the main goals of subsequent works tackling the problem. The survey~\cite{puglisi2007taxonomy} gives a good overview of the main suffix sorting techniques developed in the two decades following the introduction of suffix arrays. 
Relevant to our work are the results of Franceschini and Muthukrishnan~\cite{franceschini2007place}, Goto~\cite{goto2017optimal}, and Li et al. \cite{li2016optimal}. The authors of these papers showed that suffix sorting is possible within the same space of the text and the final suffix array, that is, \emph{in place}.
Parallel to the study of techniques to sort all suffixes of a text, several authors started considering the problem of efficiently sorting only a \emph{subset} of $b$ text's suffixes~\cite{karkkainen1996sparse, karkkainen2014faster,gawrychowski2017sparse,fischer2016deterministic,bille2013sparse,bille2015longest,karkkainen2006linear}, a fundamental step in the construction of compressed and sparse text indexes~\cite{karkkainen1996sparse} and space-efficient compression algorithms. Very recently, Gawrychowski and Kociumaka~\cite{gawrychowski2017sparse} gave the first optimal time-and-space solution to the problem, showing that $\bigO(b)$ working space and $\bigO(n)$ running time are achievable with a Monte Carlo algorithm (they also consider a Las Vegas algorithm with higher running time). Interestingly, to date no in-place (i.e. $\bigO(1)$ working space) and sub-quadratic (i.e. $o(n\cdot b)$ time) algorithm is known for the general sparse suffix sorting problem. Such an algorithm should take as input a text $T$ and an array $S$ of $b$ text positions, and suffix-sort $S$ using $\bigO(1)$ words of working space on top of $T$ and $S$.
The hardness of this problem resides --- arguably --- in its generality: since the $b$ text positions to be sorted can be arbitrarily distributed in the text, it seems hard to devise prefix-doubling or recursive techniques such as the ones existing for the full suffix array or for equally-sampled sparse suffix arrays. 
It is well known (see, e.g.~\cite{bille2013sparse}) that the Longest Common Extension problem (LCE) --- that is, to find the length of the longest common prefix between any two text suffixes --- is closely related to the suffix sorting problem. In this paper, we exploit this relation and give the first in-place solution to the sparse suffix sorting problem. We tackle the problem from a different standpoint with respect to the in-place suffix sorting algorithms described in~\cite{franceschini2007place,goto2017optimal,li2016optimal}: instead of inducing in-place the ordering of  $S$ starting from the already-defined ordering of a subset of $S$, we devise a strategy to \emph{replace in-place} the text with an implicit data structure (i.e. of the same size of the text plus a constant number of memory words) supporting fast LCE and text access queries. To achieve this result, we first show how to replace the text with the Karp-Rabin fingerprints~\cite{karp1987efficient} of a subset of its prefixes. This text representation supports fast computation of the fingerprint of any text substring and optimal-time text extraction, and therefore can be used to answer LCE queries with a simple binary search strategy.
By slightly modifying Karp-Rabin's technique, we are then able to compress these fingerprints to exactly the text's size (plus a constant number of memory words). Our transformation is fully reversible, implying in particular that the text can be reconstructed (and efficiently accessed) from it. 
More in detail, let $w$ be the memory word size (in bits), $T\in\{0,\dots, \sigma-1\}^n$ be the input text stored in $n\lceil\log\sigma\rceil$ bits, and $\ell$ be the result of the LCE query. We show that $T$ can be replaced in-place and $\bigO(n)$ expected time with a data structure of size $n\lceil\log\sigma\rceil + \bigO(w)$ bits  supporting $\bigO(\log^2\ell)$-time LCE queries with high probability and optimal $\bigO(m\log\sigma/w)$-time extraction of any length-$m$ text substring. LCE queries can be speeded up to $\bigO(\log\ell)$ time using $\bigO(\log n)$ additional words of space. 
By allowing $\bigO(n)$ additional words of working space and $\bigO(n\log n)$ expected time during construction, our structure can be made \emph{deterministic}, i.e. the returned LCE values are always correct. 

As a first immediate result, we obtain the first solution for the LCE problem that uses constant space on top of the text while supporting polylogarithmic-time LCE queries (see~\cite{bille2015longest,bille2014time,tanimura2016, nishimoto2015dynamic,inenaga2015faster, tanimura2017small} for other space-efficient solutions). Our structure is strictly more powerful than a plain text representation and is of independent interest. 
Interestingly, our data structures permit to circumvent a lower bound that holds in the read-only model for LCE structures occupying at least linear space. If we are not allowed to modify the text (but only store additional structures on top of it), then the relation $s(n)t(n) \in \Omega(n\log n)$ must hold, where $s(n)\in\Omega(n)$ and $t(n)$ are the space (in bits) used on top of the text and the time for answering LCE queries, respectively~\cite{kosolobov2017tight}. Both our bounds satisfy, instead, $s(n)t(n) \in \bigO(\log^3 n)$. This is possible because (i) we use $\bigO(\mathtt{polylog}(n))$ space on top of the text, which is in contrast with the requirement $s(n)\in\Omega(n)$ of~\cite{kosolobov2017tight}, and (ii) we allow the text to be overwritten.
Our LCE structure can be directly used to solve in-place the sparse suffix sorting problem. We present an algorithm sorting an arbitrary subset of $b$ text suffixes in-place and $\bigO(n + b\log^2 n)$ expected time. The algorithm returns the correct result with high probability. 

The second problem we consider is that of building in-place the Longest Common Prefix array (LCP), that is, the array storing the lengths of the longest common prefixes between lexicographically adjacent text suffixes.
As SA construction algorithms, LCP construction algorithms have been the subject of several research articles (see~\cite{puglisi2008space, gog2011fast} and references therein). As opposed to suffix arrays, in-place algorithms for building the LCP array have been considered only very recently~\cite{LOUZA201714}. The time-gap between solutions for building in-place the (full) SA and the LCP is considerable, as the fastest known algorithm for the latter problem runs in $\bigO(n^2)$ time~\cite{LOUZA201714}. The second contribution of this paper is the first sub-quadratic in-place LCP array construction algorithm. Our algorithm uses constant space on top of $T$ and LCP, runs in $\bigO(n\log n)$ expected time on alphabets of size $n^{\bigO(1)}$, and returns always the correct result. On more general alphabets whose elements fit in a constant number of memory words, we provide an in-place LCP construction algorithm running in $\bigO(n\log^2 n)$ expected time. 
To achieve these results, we show how to de-randomize our LCE data structure in $n$ words of space by applying multiple rounds of in-place integer sorting to the text's positions using Karp-Rabin's fingerprints as comparison values. We then build in-place the suffix array with existing techniques~\cite{franceschini2007place} and convert it to the LCP array using our LCE data structure. $\bigO(n\log n)$ running time is achieved  using in-place radix sorting~\cite{radixsort} and compressing a portion of the suffix array to get the extra space for supporting fast LCE queries.  
Applying the same techniques to the sparse suffix array, we obtain the first solution to the in-place sparse LCP array (SLCP) construction problem. Our algorithm replaces a set $S$ of $b$ text position with the SLCP array relative to the SSA of $S$ in $\bigO(n + b\log^2 n)$ expected time. The algorithm returns the correct answer with high probability.

To conclude, we consider the \emph{suffix selection} problem: to return the $i$-th lexicographically smallest text suffix. It is known that this problem can be solved in optimal $\bigO(n)$ time and $\bigO(n)$ words of space on top of the text~\cite{franceschini2007optimal}. Considering that the output consists of only one text position, this solution is far from being space-efficient. In this paper, we present the first in-place solution for the suffix selection problem: our algorithm runs in $\bigO(n\log^3 n)$ expected time, uses only constant space on top of the text, and returns the correct result with high probability. All our algorithms work under the assumption that the text is re-writable. After execution, we restore the text in its original form.

\section{Preliminaries}

We work under the following model. We assume our input $T$ to be a text of length $n$ drawn from an integer alphabet $\Sigma = \{0,\dots,\sigma-1\}$, stored using $\lceil\log \sigma\rceil$ bits per character. 
We moreover assume that $T$ is re-writable. After execution, our algorithms restore $T$ in its original form.
Even though these might seem to be strong requirements, we note that the same assumptions are often (implicitly) made when dealing with in-place algorithms. For example, the assumption that the input takes $\lceil\log \sigma\rceil$ bits per character is equivalent to the one made when sorting in-place $n$ integers of $\lceil\log \sigma\rceil$ bits each, where $\sigma$ is the largest integer in the array. It is known that the integers could be represented so that they use overall $\lceil n\log_2 \sigma\rceil + \bigO(1)$ bits~\cite{dodis2010changing}, and better representations such as arithmetic or prefix-free encodings achieve even compressed space. Standard in-place sorting algorithms need to assume efficient access to the array, which translates (implicitly) to assumptions on the input format.
Generalizing our solutions to more efficient encodings represents therefore a first line of improvement over the results described in this paper. As far as the requirement of re-writable text is concerned, 
the usual definition of in-place algorithm is ``an algorithm which transforms the input into the output using constant additional working space'' (where ``constant'' is measured in computer memory words). In the case of sparse suffix sorting, the input is represented by the text $T$ and the array of positions $S$, and the output consists of $T$ and the lexicographically-sorted $S$ (or even just $S$).
It is also true that, in our case, we cannot exclude that the same results described in this paper could be obtain re-using only the space of $S$ (as done in~\cite{franceschini2007place} for the full suffix array). Also this represents a possible line of improvement over our work.

$w\geq 1$ is the memory word size (in bits). 
In our proofs we assume $n \in \omega(1)$ (clearly, if $n \in \bigO(1)$ then all considered problems can be trivially implemented in constant time), 
$\log n \leq w$, and $\lceil \log\sigma\rceil \leq w$. 
Since we make use only of integer additions, multiplications, modulo, and bitwise operations (masks, shifts), we assume that we can simulate a memory word of size $w'=c\cdot w$ for any constant $c$ with only a constant slowdown in the execution of these operations. Subtractions, additions and multiplications between $(c\cdot w)$-bits  words take trivially constant time by breaking the operands in $2c$ digits of $w/2$ bits each and use schoolbook's algorithms (i.e. $\bigO(c)$-time addition and $\bigO(c^2)$-time multiplication). The modulo operator $a\mod q$ (with $q$ fixed) can be computed as $a\mod q = a - \lfloor a/q \rfloor\cdot q$. Computing $\lfloor a/q \rfloor$ can be done efficiently using Knuth's long division algorithm~\cite{knuthbook}.
Bitwise operations on $(c\cdot w)$-bits  words can trivially be implemented with $c$ bitwise operations between $w$-bits words.
Since we can simulate words of size $\bigO(w)$ with no asymptotic slowdown, our results are  generalizable to the case $\log n \in \bigO(w)$ and $\log\sigma \in \bigO(w)$. For some of the results described in Section \ref{sec:in-place results} we will require the stricter bound $\sigma \leq n^{\bigO(1)}$.

For a more compact notation, with $T[i,\dots,j]$ we denote both $T$'s substring starting at position $i$ and ending at position $j$, and the integer with binary representation $T[i]T[i+1]\dots,T[j]$ (each $T[i]$ being a $\lceil\log \sigma\rceil$-bits integer). If $j<i$, $T[i,\dots,j]$ denotes the empty string $\epsilon$ or the integer $0$. The use (string/integer) will be clear from the context.
$T.LCE(i,j)$ indicates the length of the longest common prefix of $T[i,\dots, n-1]$ and $T[j,\dots, n-1]$, i.e. the $i$-th and $j$-th suffixes of $T$.

$\rk:\{0,1\}^*\rightarrow [0,q-1]$ indicates the Karp-Rabin hash function~\cite{karp1987efficient} with modulo $q$ on strings from the binary alphabet $\{0,1\}$. This function is defined as 
$\rk(S) = S\mod q$, where `$S$' has to be interpreted as a binary string on the left hand-side of the equation, and as a binary number of $|S|$ digits on the right-hand side of the equation.

W.h.p. (with high probability) means with probability at least $1 - n^{-c}$ for an arbitrarily large constant $c$.
If not otherwise specified, logarithms are in base $2$.

\section{A Monte Carlo in-place LCE data structure}\label{sec:monte carlo LCE}

Our strategy to obtain a data structure supporting LCE queries follows the one described by Bille et al.~\cite{bille2014time,bille2015longest}. Our improvements over this result regard space of the structure and its construction time. First, we describe an implicit data structure supporting efficient computation of the Karp-Rabin fingerprint of any text substring. Our original idea is to \emph{replace} the text (not just augment it) with the fingerprints of a subset of its prefixes. The loss in space efficiency is avoided by repeatedly picking the Karp-Rabin modulus $q$ very close, but above, a power of two until all residues are below that power of two (thus saving 1 bit per stored fingerprint). 
Using Karp-Rabin fingerprints, LCE queries are then answered with a technique similar to the one used in~\cite{bille2014time,bille2015longest} (i.e. exponential and binary search on $\bigO(\log \ell)$ prefixes of the two text suffixes).

We start considering the binary case $\sigma=2$, and then extend the result to more general alphabets. We introduce two sources of randomness in our structure: the \emph{modulus} $q$ and the \emph{seed} $\seed$. First, we choose a random prime $q$ uniformly in the interval~\footnote{Note that we can generate uniform primes from any interval with the naive algorithm that picks a random integer from that interval, tests it for primality, and returns it if prime (repeating until a prime is found). See~\cite{fouque2014close} for more efficient methods.} $[2,2^w-1]$. We define a \emph{block size} $\tau=\lceil \log q \rceil$. Without loss of generality, we assume that $n$ is a multiple of $\tau$ (the general case can be reduced to this case by left-padding the text with $\tau-(n\mod\tau)$ bits). 
At this point, we choose uniformly a random number (the seed) $\seed$ in the interval $[0,q-1]$. $\seed$ is an integer of $\tau$ bits (after a suitable left-padding of zeros); we left-pad our binary text $T$ with $\seed$ written in binary. Clearly, $LCE(i,j)$ queries on $T$ can still be solved using the padded text $\seed T$ by simply adding $\tau$ to the arguments of LCE (i.e. solving $LCE(i+\tau,j+\tau)$ on the padded text). To improve readability, in what follows we assume that $T$ is prefixed by $\seed$ (and thus write just $T$ and $n$ instead of $\seed T$ and $n+\tau$, respectively).
Let $B, P'\in [0, q-1]^{n/\tau}$ be the arrays defined as
$$
B[i] = T[i\cdot\tau,\dots,(i+1)\cdot\tau-1],\ \ \ i=0,\dots, n/\tau - 1
$$
and
$$
\begin{array}{ccc}
P'[i] & = & \sum_{j=0}^{i} 2^{(i-j)\cdot \tau}\cdot B[j] \mod q \\
& = & \rk(T[0,\dots,(i+1)\cdot\tau-1])
\end{array}
$$
First, note that
$B[i]-q \leq (2^\tau-1)-2^{\tau-1} < 2^{\tau-1} \leq q$, 
so $\lfloor B[i] / q \rfloor \in \{0,1\}$ holds. We build a bitvector $D[0,\dots, n/\tau-1]$ defined as $D[i] = \lfloor B[i] / q \rfloor$. To simplify notation, let $P'[-1]=0$. At this point, $B$'s values can be retrieved as
$
B[i] = \left(P'[i] - 2^\tau\cdot P'[i-1] \mod q\right) + D[i]\cdot q
$.
Arrays $P'$ and $D$ take $n + n/\tau$ bits of space and replace the text in that they support the retrieval of any $B[i]$. First, we show how to compute efficiently $\rk(T[i,\dots,j])$ for any $0\leq i \leq j < n$ by using $P'$ and $D$. We then show how to reduce the space usage to $n + \bigO(w)$ bits while still being able to support constant-time text extraction and Karp-Rabin fingerprint computation.
Let $j=\lfloor i/\tau\rfloor$. Then, 
$$
\begin{array}{lll}
\rk(T[0,\dots, i]) & = & \rk(T[0,\dots, j\cdot\tau-1])\cdot 2^{i-j\cdot\tau+1}\ + \\
&& \rk(T[j\cdot\tau, \dots, i]) \mod q\\
& = & P'[j-1]\cdot 2^{i-j\cdot\tau+1}\ + \\
&& \lfloor B[j]/2^{\tau-i+j\cdot\tau-1}\rfloor \mod q
\end{array}
$$

We now have to show how to compute the fingerprint $\rk(T[i,\dots, j]),\ j\geq i$, of any text substring. This can be easily achieved by means of the equality:
\begin{equation}\label{eq: rk substring}
\begin{array}{lll}
\rk(T[i,\dots, j]) & = & \rk(T[0,\dots, j])\ - \\
&& \rk(T[0,\dots, i-1])\cdot 2^{j-i+1} \mod q
\end{array}
\end{equation}
Computing $2^e\mod q$ takes $\bigO(\log e)$ time with the fast exponentiation algorithm, therefore the computation of $\rk(T[i,\dots, j])$ takes $\bigO(\log (j-i+1))$ time with our structure for all $0\leq i \leq j < n$.

\subsection{Reducing space usage}\label{sec:reducing space}

In order to remove the $n/\tau$-bits overhead, we build an auxiliary array 
$
S = \langle i\ :\ \lfloor P'[i]/2^{\tau-1}\rfloor = 1,\ i=0,\dots, n/\tau-1 \rangle
$
storing all $i$'s such that the most significant bit of $P'[i]$ is equal to 1. At this point, we replace $P'$ with an array $P$ of $n/\tau$ integers of $(\tau-1)$ bits each defined as
$
P[i] = P'[i] \mod 2^{\tau-1},\ \ \ i=0,\dots,n/\tau-1
$,
i.e. we remove the most significant bit from each $P'[i]$. $P$ takes $n\cdot(\tau-1)/\tau$ bits of space. Clearly, by using $P$ and $S$ we can retrieve any $P'[i]$ in $\bigO(|S|)$ time with a simple linear scan on $S$. The main idea, at this point, is to  choose  the prime $q$ in such a way that the expected size of $S$ becomes constant (in fact: equal to zero). 

We reverse our strategy. We first choose a block size $\tau\in\Theta(w)$, pick a uniform prime $q$ such that $\lceil\log q\rceil = \tau$, and then choose a uniform seed $\seed$ in $[0,q-1]$. 
Operations on integers of size $\tau\in\Theta(w)$ can still be performed in constant time. 
In Sections \ref{sec:det space} and \ref{sec:Monte Carlo structure} we show how to choose $\tau$.
The key point is that each $P'[i]$ is a uniform random variable taking values in the range $[0,q-1]$. To prove this statement, note that $P'[i]$ can be written as
$P'[i] = \seed\cdot 2^{i\cdot \tau} + \bar t_i \mod q$,
where $\bar t_i =\rk(T[\tau,\dots, (i+1)\cdot \tau-1])$. Let $\mathcal P(P'[i]=x)$, $x<q$, be the probability that $P'[i]$ is equal to $x$.  Then, for any $x<q$,
\begin{equation}\label{eq:unif P'}
\begin{array}{lll}
\mathcal P(P'[i]=x) &=& \mathcal P(\seed\cdot 2^{i\cdot \tau} + \bar t_i \equiv_q x) \\
&=&  \mathcal P(\seed \equiv_q (x - \bar t_i) \cdot 2^{-i\cdot \tau} )\\
&=& 1/q 
\end{array}
\end{equation}

The fact that $q$ is prime guarantees the existence of the inverse of $2^{i\cdot \tau}$ modulo $q$. Let $x<q$. Equation \ref{eq:unif P'} implies, in particular, that 
\begin{equation}\label{eq:P' smaller than} 
\mathcal P(P'[i]<x) = x/q
\end{equation}

Let $\bar 1_i\in\{0,1\}$ be the indicator random variable taking value $1$ iff the most significant bit of $P'[i]$ is equal to 1. Equation \ref{eq:P' smaller than} implies that $\bar 1_i$ has a Bernoullian distribution with success probability $p=(q-2^{\tau-1})/q$.
We want this probability to be at most $1/n$ in order to get an expected constant size for $S$. By solving $(q-2^{\tau-1})/q \leq 1/n$ and by adding the constraint $\lceil\log q\rceil = \tau$, we obtain that the interval $\mathcal Z$ from which we have to uniformly pick $q$ in order to satisfy both constraints is
\begin{equation}\label{eq:interval}
\mathcal Z = \left[2^{\tau-1}, 2^{\tau-1}\left(\frac{n}{n-1}\right)\right]
\end{equation}
At this point, $S$'s expected size $E[|S|]$ can be computed as
\begin{equation}\label{eq:exp S size}
\begin{array}{lll}
E[|S|] &=& E\left[ \sum_{i=0}^{n/\tau-1} \bar 1_i \right]\\ 
&=&  \sum_{i=0}^{n/\tau-1} E[\bar 1_i] \\
&=&  \frac{n}{\tau} E[\bar 1_i]\\
&=&  \frac{n}{\tau}\cdot \frac{q-2^{\tau-1}}{q} \\
&\leq& \frac{n}{\tau}\cdot \frac{1}{n} =  \frac{1}{\tau}
\end{array}
\end{equation}

Let $b=\lceil\log\sigma\rceil$. On a general alphabet size such that $b\leq w$, the text is processed as a binary sequence of $nb$ bits, and the validity of the above results is preserved by substituting $n$ with $nb$ in Equation \ref{eq:interval}. In particular, $\mathcal Z$'s size becomes $|\mathcal Z| = 2^{\tau-1}/(nb-1)\geq 2^{\tau-1}/(nb) = 2^{\tau-1-\log n-\log b}$. Since we assume $w\geq \log n$ and $w\geq b \geq \log b$, we get the lower bound 
\begin{equation}\label{eq: Z lower bound}
|\mathcal Z| \geq 2^{\tau-1-2w}
\end{equation}

Let $\pi(x)$ denote the number of primes smaller than $x$. Let moreover $A = 2^{\tau-1}$ and $H = 2^{\tau-1-2w}$ be the smallest element contained in $\mathcal Z$ and the lower bound for $|\mathcal Z|$ stated in Equation \ref{eq: Z lower bound}, respectively. Our aim in this paragraph is to compute a lower bound for the number $z_p = \pi(A+H) - \pi(A)$ of primes contained in $\mathcal Z$. This will be needed later in order to compute the collision probability of our hash function. 
Note that the Prime Number Theorem can be applied to solve this task only if $H \geq A\cdot c$, for some fixed $c>0$, so we cannot use it in our case. Luckily for us, Heath-Brown~\cite{heath1978differences} proved (see also~\cite{maier1985primes}) that, if $H$ grows at least as quickly as $A^{7/12}$, then $\pi(A+H)-\pi(A) \sim H/\log_e A$  (for $A\rightarrow \infty$. $e$ is the natural logarithm base). Solving $H \geq A^{7/12}$ we get the constraint
\begin{equation}\label{constraint2}
\tau \geq (24/5)w+1 
\end{equation}
Later we will show how to choose $\tau$ (keeping (\ref{constraint2}) in mind). Heath-Brown's theorem gives us $z_p \geq \frac{H}{\log_e A} = \frac{2^{\tau-1-2w}}{\log_e(2^{\tau-1})} \geq \frac{2^{\tau-1-2w}}{(\tau-1)} \geq \frac{2^{\tau-1-2w}}{\tau}$.

\subsection{Deterministic space}\label{sec:det space}

With the above strategy, our structure takes $n\lceil\log\sigma\rceil + \bigO(w)$ bits of space with high probability only. We can assure that the space is \emph{with certainty} $n\lceil\log\sigma\rceil + \bigO(w)$ bits by picking multiple random pairs $q,\seed$ as described above and re-building the structure until this requirement is satisfied (i.e. we move the randomness into the construction algorithm). Our goal in this section is to compute the expected number $R$ of pairs $q,\seed$ we have to randomly pick before obtaining an empty $S$.
Note that $E[|S|] = \sum_{k > 0} k\cdot \mathcal P(|S|=k)$ and $\mathcal P(|S|> 0) = \sum_{k > 0} \mathcal P(|S|=k)$, so $\mathcal P(|S|> 0) \leq E[|S|]$ holds. From Equation \ref{eq:exp S size}, $E[|S|] \leq 1/\tau$, therefore $\mathcal P(|S|> 0) \leq 1/\tau$.
This yields $\mathcal P(|S| = 0) \geq 1 - 1/\tau$. We choose 
\begin{equation}\label{constraint1}
\tau = cw \geq c\log n
\end{equation}
for any constant $c\geq 1$ fixed at construction time, so the above probability is at least $1-1/\log n$. Later we will show how to choose $c$ keeping in mind also Constraint (\ref{constraint2}). Finally, since we assume $n \in \omega(1)$, then $\log n\geq 2$ and we obtain $\mathcal P(|S|=0) \geq 0.5$.
Given that we repeat the construction of our structure as long as $|S|>0$ holds, the number $R$ of times we repeat the construction is a geometric random variable with success probability $p=\mathcal P(|S|=0) \geq 0.5$, and has therefore expected value $1/p \leq 2$. Note that, since $\tau\in\Theta(w)$ and $\lceil\log\sigma\rceil \leq w$, arrays $P$ and $D$ have $\bigO(n)$ entries each. We obtain the following Lemma:

\begin{lemma}\label{lemma: PDS}
	In $\bigO(n)$ expected time we can build arrays $P$, $D$, and $S$ taking overall $n\lceil\log\sigma\rceil +\bigO(w)$ bits of space and supporting the computation of any $B[i]$ and $P'[i]$ in constant time.
\end{lemma}

\subsection{Monte Carlo LCE data structure}\label{sec:Monte Carlo structure}

On a binary alphabet, we can easily answer $LCE(i,j)$ by comparing $\rk(T[i,\dots, i+k])$ with $\rk(T[j,\dots, j+k])$ for $\bigO(\log n)$ values of $k$ with binary search. We can furthermore improve this query time by performing an exponential search before applying the binary search procedure. We compare $\rk(T[i,\dots, i+k])$ with $\rk(T[j,\dots, j+k])$ for $k=2^0, 2^1, 2^2, \dots$ until the two fingerprints differ. Letting $\llen = LCE(i,j)$, this procedure terminates in $\bigO(\log\llen)$ steps. We then apply the binary search procedure described above on the interval of size $\bigO(\llen)$ obtained with the exponential search. Each exponential and binary search step take $\bigO(\log\ell)$ time (from the fast exponentiation algorithm). 

On a more general alphabet, each character takes $b=\lceil\log\sigma\rceil\in\bigO(w)$ bits, and our structure is therefore built over a binary text $T'$ of length $n\cdot b$. We can make the query time alphabet-independent as follows. First of all, while computing $T.LCE(i,j)$ we perform exponential and binary searches by comparing $\rk(T'[i\cdot b,\dots, (i+k)\cdot b])$ with $\rk(T'[j\cdot b,\dots, (j+k)\cdot b])$, i.e. we compare $T'$ substrings starting and ending at character boundaries. This reduces the number of steps to be performed from $\bigO(\log(\ell\cdot b))$ to $\bigO(\log \ell)$. At this point, note that each step requires the computation of $2^{t\cdot b}\mod q$ with the fast exponentiation algorithm, $t\in\bigO(\ell)$ being the length of the two compared substrings ($\bigO(\log(\ell\log\sigma))$ time). Since $b$ is a common factor in all exponents, we can pre-compute $Y=2^b\mod q$ and---at each step---compute $Y^t\mod q$ instead of $2^{t\cdot b}\mod q$ with the fast exponentiation algorithm. This reduces the number of steps of the exponentiation algorithm to $\bigO(\log \ell)$. 
Finally, note that extracting text corresponds to reading array $B$ ($\tau\in \Theta(w)$ bits of the text per $B$ element). 

Plugging the lower bound---computed in Section \ref{sec:reducing space}---for the number of primes contained in $\mathcal Z$ inside a standard analysis for Karp-Rabin collision probability, we can prove the following:

\begin{lemma}\label{lemma:wong LCE}
	If we choose $\tau \geq (9+c)w$ for an arbitrarily large constant $c$, then the probability that our LCE structure returns a wrong result is upper-bounded by $n^{-c}$.
\end{lemma}
\begin{proof}
We start the analysis from the binary case $\sigma=2$. Let $C$ be the random variable denoting the number of pairs $\langle X,Y\rangle$ of equal-length $T$ substrings ($|X|=|Y|$) that generate a collision, i.e. $X\neq Y$ and $\rk(X)=\rk(Y)$. Our goal is to compute an upper bound for the probability $\mathcal P(C>0)$. Clearly, $\mathcal P(C>0)$ is an upper bound to the probability of computing a wrong LCE with our structure.
Let $X_i^k$ denote $T$'s substring of length $k$ starting at position $i$. There is at least one collision ($C>0$) iff $X_i^k \equiv_q Y_j^k$ for at least one pair $X_i^k \neq Y_j^k$, i.e. iff $q$ divides at least one of the numbers $|X_i^k - Y_j^k|$ such that $X_i^k \neq Y_j^k$. Since $q$ is prime, this happens iff $q$ divides their product
$z = \prod_{k=1}^{n-1} \prod_{i,j: X_i^k \neq Y_j^k} |X_i^k - Y_j^k|$.
Since each $|X_i^k - Y_j^k|$ has at most $n$ binary digits and there are no more than $n^2$ such pairs for every $k$, we have that $z< 2^{n^4}$. It follows that there cannot be more than $n^4$ distinct primes dividing $z$. 

Let $b=\lceil\log \sigma\rceil$. On a more general alphabet with $b \leq w$, each $|X_i^k - Y_j^k|$ has at most $nb$ binary digits, and we obtain $z< 2^{n^4\cdot b}$. It follows that there cannot be more than $n^4\cdot b$ distinct primes dividing $z$.  The probability of uniformly picking a prime $q\in\mathcal Z$ dividing $z$ is therefore upper bounded by $n^4\cdot b/z_p$, where $z_p$ is a lower bound on the number of primes contained in $\mathcal Z$, see Section \ref{sec:reducing space}. Recall that $z_p \geq 2^{\tau-1-2w}/\tau$, so $n^4\cdot b/z_p \leq n^4\cdot b\cdot \tau/2^{\tau-1-2w} = 2^{4\log n+\log b+\log\tau+1+2w-\tau}$. We choose 
\begin{equation}\label{constraint3}
\tau = (9+c)w
\end{equation}
for an arbitrarily large constant $c$. Being $w\in\omega(1)$ (because $n\in\omega(1)$ and $w\geq \log n$), we assume\footnote{Note that this inequality always holds after simulating a memory word of size $w'=dw$ for a sufficiently large constant $d$.} $w \geq \log\tau = \log(9+c)+\log w$. We obtain $\tau = (9+c)w \geq   (4+c)\log n + \log b + \log\tau + 1 + 2w$, therefore $n^4\cdot b/z_p\leq 2^{-c\log n}=n^{-c}$. Note that this choice of $\tau$ satisfies constraints (\ref{constraint2}) and (\ref{constraint1}). This leads to:
$$
\mathcal P(wrong\ LCE) \leq \mathcal P(C>0) \leq n^{-c}
$$
for an arbitrarily large constant $c$.
\qed
\end{proof}

Lemma \ref{lemma:wong LCE} leads to our first core result:

\begin{theorem}\label{th:MC structure 1}
	In $\bigO(n)$ expected time we can build a data structure of $n\lceil \log\sigma\rceil + \bigO(w)$ bits of space  supporting extraction of any length-$m$ text substring and LCE queries w.h.p. in 
	$\bigO\left(m\log\sigma/w \right)$ and $\bigO(\log^2\ell)$ worst-case time, respectively.
\end{theorem}

Let $b=\lceil\log\sigma\rceil$ and let $T'\in\{0,1\}^{n\cdot b}$ be the concatenation of $T$'s characters written in binary. We can avoid the overhead introduced by the fast exponentiation algorithm by pre-computing and storing (in $\bigO(\log n)$ words) values $z_i = 2^{b\cdot 2^i}\mod q,\ i=0,\dots,\lfloor\log n\rfloor$ and always comparing text substrings whose length is a power of two during binary search:

\begin{theorem}\label{th:MC structure 2}
	In $\bigO(n)$ expected time we can build a data structure 
	taking $n\lceil \log\sigma\rceil + \bigO(w\log n)$ bits of space and supporting extraction of any length-$m$ text substring and LCE queries w.h.p. in 
	$\bigO\left(m\log\sigma/w \right)$ and $\bigO(\log\ell)$ worst-case time, respectively.
\end{theorem}
\begin{proof}
	First, note that $z_0=2^b\mod q$ and $z_{i+1} = (z_i)^2\mod q$, so the values $z_i$ can be pre-computed in $\bigO(\log n)$ time. Let the notation $\langle i,j,e,k\rangle$, with $0\leq i,j,e,k<n$ and $e<k$, denote that we already verified (w.h.p.) that $T[i,\dots, i+e-1] = T[j,\dots, j+e-1]$ and $T[i,\dots, i+k-1] \neq T[j,\dots, j+k-1]$. We use this notation to indicate the state of a binary search step, and start from state $\langle i,j,0,n-j \rangle$ (we assume for simplicity that $T[i,\dots, i+(n-j)-1] \neq T[j,\dots, n-1]$; otherwise, $LCE(i,j)=n-j$). We use a modified version of Equation \ref{eq: rk substring} by adding a parameter (exponential $E$) to the Karp-Rabin hash function:
	\begin{equation}\label{eq: rk substring1}
	\begin{array}{lll}
	\rk'(T'[i,\dots, j],E) & = & \rk(T'[0,\dots, j])\ - \\
	&& \rk(T'[0,\dots, i-1])\cdot E \mod q
	\end{array}
	\end{equation}
	Note that $\rk(T'[i,\dots, j]) = \rk'(T'[i,\dots, j],2^{(j-i+1)\cdot b})$. At binary search step $\langle i,j,e, k\rangle$ we still have to compare the last $l = k-e$ characters of $T[i,\dots,i+k-1]$ and $T[j,\dots,j+k-1]$. We split each of these two substrings in the left part of length $l' = 2^{\lfloor\log(l/2)\rfloor}$ (i.e. the closest power of 2 smaller than or equal to $l/2$) and the right part of length $l-l'$. Note that value $2^{l'\cdot b}\mod q = z_{\log l'} = z_{\lfloor\log(l/2)\rfloor}$ has been pre-computed, so we can compute and compare in constant time the two values
	$$
	\begin{array}{l}
	\rk(T'[(i+e)\cdot b,\dots,(i+e+l'-1)\cdot b]) = \\
	\rk'(T'[(i+e)\cdot b,\dots, (i+e+l'-1)\cdot b], z_{\lfloor\log(l/2)\rfloor})
	\end{array}
	$$
	and
	$$
	\begin{array}{l}
	\rk(T'[(j+e)\cdot b,\dots,(j+e+l'-1)\cdot b]) = \\
	\rk'(T'[(j+e)\cdot b,\dots, (j+e+l'-1)\cdot b], z_{\lfloor\log(l/2)\rfloor})
	\end{array}
	$$
	If the two values differ, then we recurse on $\langle i,j,e,e+l'\rangle$. If the two values are equal, then we recurse on $\langle i,j,e+l',k\rangle$. Note that we always compare (fingerprints of) strings whose lengths are powers of two. This will be crucial in the next section in order to efficiently de-randomize our structure. Since $l/4 < l'\leq l/2$, this binary search procedure terminates in $\bigO(\log n)$ steps, each taking constant time. As done in the previous section, we can perform an exponential search before the binary search in order to reduce the size of the binary search interval from $\bigO(n)$ to $\bigO(\llen)$. Note that with our sampling $z_i$ it is straightforward to implement each exponential search step in constant time. Note moreover that values $z_i$ need to be explicitly stored for the binary search as this step might need access to a $z_i$ with arbitrarily small $i$ (and, while we can quickly compute $z_i$ from $z_{i-1}$, the opposite is not true). \qed
\end{proof}

Note that, in Theorem \ref{th:MC structure 2}, we might as well allocate at most $\bigO(\log \ell)$ words during query time on top of the $n\lceil \log\sigma\rceil$ bits of the main structure. 
Note also that, even if we assumed $\lceil \log\sigma\rceil \leq w$, it is easy to see that our results are valid also for the more general case $\log \sigma \in \bigO(w)$ by simulating a larger word.

\subsection{In-place construction algorithm}\label{sec:in-place}

In this section we show that our data structure can be built \emph{in-place}, i.e. we can replace the text with the data structure and use only $\bigO(1)$ memory words of extra space during construction. 

We first consider the binary case $\sigma=2$. First, we pick $\tau$, $q$, $\seed$ as described in the previous sections. We consider the text as a sequence $B[0,\dots,n/\tau-1]$ of integers in the range $[0,2^\tau-1]$ (again, we assume for simplicity that $\tau$ divides $n$), and, for $i=1,\dots, n/\tau-1$, we 
compute $P'[i]$ and discard $P'[i-1]$. If the most significant bit of $P'[i]$ is equal to 1 at any construction step $i$, then we pick another random pair $q$, $\seed$ and repeat the process from the beginning. 
From Section \ref{sec:det space}, we need to pick at most $\bigO(1)$ pairs in the expected case before the most significant bits of all $P'[i]$ are equal to 0 (and, in particular, $P' = P$).
At this point, we scan one last time the text and, for $i=1,\dots, n/\tau-1$, we replace the most significant bit of $B[i]$ with the bit $D[i]$, and the remaining $\tau-1$ bits with the value $P[i]$ (clearly, this allows to retrieve any $P[i]$ and $D[i]$ in constant time). 
Overall we spend $\bigO(n/w)$ time. On a general alphabet with $\lceil\log\sigma\rceil \in\bigO(w)$, we build the structure on the binary representation of the text and the construction algorithm terminates therefore in expected optimal $\bigO(\frac{n\log\sigma}{w})$ time (assuming that the input text is already packed) while taking only $\bigO(1)$ words on top of the space of the text.

As a final remark, note that the above process can be easily reverted to restore the text. It is easy to see that $B[i]$ can be computed in constant time using $P[i]$, $P[i-1]$, and $D[i]$ (recall that $B[i]$ has been replaced with $P[i]$ and $D[i]$ during construction), so we can restore the text in $\bigO(n/w)$ time using $\bigO(1)$ words of working space.

\section{Exact LCE queries}\label{sec:deterministic structure}

The aim of this section is to show how to make sure that our data structure always returns the correct result. We achieve this by moving the randomization in the construction process. We start by proving three lemmas solving with different space/time trade-offs the problem of checking whether $\rk$ generates collisions over a specific subset of text substrings. The first lemma is due to Bille et al.~\cite{bille2014time}:

\begin{lemma}\label{th:derand1}
	In $\bigO(n\log n)$ expected time and $\bigO(n)$ words of space we can check whether $\rk$ is collision-free over all pairs of substrings of $T$ having the same length $k=2^e$, for all $0\leq e\leq \log n$.
\end{lemma}
\begin{proof}
	The idea is to check the property on strings of length $2^e$ by using the already-checked property on strings of length $2^{e-1}$.
	First, we build a data structure supporting the computation of $\rk(T[i,\dots,j])$ in constant time. 
	To this end, we can use the structure described Section \ref{sec:in-place}, augmented with $n$ words storing values $2^{\lceil\log n\rceil\cdot i}\mod q$, $i=0, \dots, n-1$ (to guarantee constant-time retrieval of powers of 2 modulo $q$).
	We start with $e=0$ and repeat $\log n+1$ times the following procedure, each time incrementing $e$ by one.
	We use a hash table $\mathcal H$ of size $\bigO(n)$ with associated hash function $h:[0,q-1]\rightarrow [0,|\mathcal H|-1]$ mapping Karp-Rabin fingerprints of length-$2^e$ text substrings to numbers in $[0,|\mathcal H|-1]$. $\mathcal H$ can be implemented with linear probing in order to guarantee expected constant-time operations (see, e.g.~\cite{goodrich2014algorithm}).
	Each entry of $\mathcal H$ is associated with a list of integers.
	We scan $T$ left-to-right and, for each $0\leq i \leq n-2^e$, append the value $i$ at the end of the list $\mathcal H[h(\rk(T[i,\dots,i+2^e-1]))]$. Then, for each $0\leq t < |\mathcal H|$, we check that all $i_1,i_2\in \mathcal H[t]$ are such that $T[i_1,\dots,i_1+2^e-1]=T[i_2,\dots,i_2+2^e-1]$. This task can be performed in $\bigO(|\mathcal H[t]|)$ time as follows. 
	
	Let $\mathcal H[t]=\langle i_0,\dots, i_{d-1} \rangle$.
	We only need to perform $d-1$ comparisons $T[i_j,\dots,i_j+2^e-1]=T[i_{j+1},\dots,i_{j+1}+2^e-1]$ for $0\leq j <d$.
	If $e=0$, then each comparison takes constant time and can be done by simply accessing the text. If $e>0$, then $T[i_j,\dots,i_j+2^e-1]=T[i_{j+1},\dots,i_{j+1}+2^e-1]$ holds if and only if both $\rk(T[i_j,\dots,i_j+2^{e-1}-1])=\rk(T[i_{j+1},\dots,i_{j+1}+2^{e-1}-1])$ and $\rk(T[i_j+2^{e-1},\dots, i_j+2^e-1])=\rk(T[i_{j+1}+2^{e-1},\dots, i_{j+1}+2^e-1])$ hold (constant time by using our structure to compute any $\rk(T[i',\dots,j'])$). Note that we already verified that $\rk$ is collision-free over $T$'s substrings of length $2^{e-1}$, so both checks never fail. All lists in $\mathcal H$ store overall $n-2^e+1$ elements, therefore the procedure terminates in $\bigO(n)$ expected time. Since we have to repeat this for every integer $0 \leq e\leq \log n$, the overall expected time is $\bigO(n\log n)$. \qed
\end{proof}

This bound can be improved by replacing hashing with in-place integer sorting:

\begin{lemma}\label{th:derand2}
	In $\bigO(n \log^2n)$ expected time and $n$ words of space (on top of $T$) we can check whether $\rk$ is collision-free over all pairs of substrings of $T$ having the same length $k=2^e$, for all $0\leq e\leq \log n$.
\end{lemma}
\begin{proof}
	First, we build in-place and $\bigO(n)$ time the in-place data structure supporting the computation of $\rk(T[i,\dots,j])$ in $\bigO(\log n)$ time described in Section \ref{sec:in-place}.
	
	For $e=0,\dots,\lfloor\log n\rfloor$ we repeat the following procedure. We initialize an array (text positions) $A[0,\dots,n-2^e]$ with $A[i]=i$. We use any $\bigO(n\log n)$ in-place comparison-sorting algorithm to sort $A$ according to length-$2^e$ fingerprints, i.e. using the ordering $\prec$ defined by $A[i]\prec A[j]$ iff $\rk(T[A[i], \dots,A[i]+2^e-1]) < \rk(T[A[j], \dots,A[j]+2^e-1])$. At this point, we scan $A$ and, for every pair of adjacent $A$'s elements $A[i], A[i+1]$, if $\rk(T[A[i], \dots,A[i]+2^e-1]) = \rk(T[A[i+1], \dots,A[i+1]+2^e-1])$, then we check deterministically that the two substrings $T[A[i],\dots, A[i]+2^e-1]$ and $T[A[i+1],\dots, A[i+1]+2^e-1]$ are indeed equal with the same strategy used in the proof of Theorem \ref{th:derand1} (i.e. we compare the fingerprints of their two halves of length $2^{e-1}$, or we just access the text if $e=0$). Finally, we free the memory allocated for $A$.
	
	Analysis. For every $e\leq \log n$ we sort $A$ ($\bigO(n\log n)$ comparisons). Note that fingerprints have all the same length $2^e$, so we only need to pre-compute value $2^{b\cdot 2^e}\mod q$, with $b=\lceil\log\sigma\rceil$, in order to support fingerprint computation in constant time. Each comparison in the sorting algorithm requires the computation of two fingerprints and takes therefore constant time. We moreover need value $2^{b\cdot2^{e-1}}\mod q$ to perform the deterministic collision checks. Since $2^{b\cdot2^e}\mod q$ can be computed in constant time from $2^{b\cdot2^{e-1}}\mod q$, we need to reserve only two memory words for this sampling of powers of 2 modulo $q$ (updating these two values every time $e$ is incremented).	
	\qed
\end{proof}

If we limit the word size to $w \in \Theta(\log n)$, then we can use in-place radix sorting~\cite{radixsort} to improve upon the above result:

\begin{lemma}\label{th:derand3}
	If $w\in\Theta(\log n)$, then in $\bigO(n\log n)$ expected time and $n$ words of space (on top of $T$) we can check whether $\rk$ is collision-free over all pairs of substrings of $T$ having the same length $k=2^e$, for all $0\leq e\leq \log n$.
\end{lemma}
\begin{proof}
	We need to show that we can sort in-place (i.e. $n\log n$ bits of space) and $\bigO(n)$ time text positions $i=0,\dots, n-2^e$ using as comparison values $\rk(T[i, \dots, i+2^e-1])$. Then, we plug this sorting procedure in the proof of Theorem \ref{th:derand2} to obtain the claimed bounds of the theorem.
	
	A Karp-Rabin fingerprint takes $\tau = \bigO(w) = \bigO(\log n)$ bits of space. Let $c'$ be a constant such that $\tau \leq c'\log n$.
	Let $x_i = \rk(T[i,\dots, i+2^e-1])i$ be the concatenation of the fingerprint of $T[i,\dots, i+2^e-1]$ and of position $i$ written in binary. $x_i$ takes $(c'+1)\log n$ bits of space (if less, left-pad with zeros). Note that Karp-Rabin Fingerprints can be computed in constant time using our Monte Carlo structure as the power $2^e\mod q$ is fixed. We store $x_0, \dots x_{n/(c'+1)-1}$ in an array $A$ taking $n\log n$ bits of space. We sort $x_0, \dots x_{n/(c'+1)-1}$ in-place and  $\bigO(n)$ time using in-place radix sort~\cite{radixsort}. Then, we compact $A$ by replacing each $x_i$ with the integer $i$. As a result,  the first $n/(c'+1)$ entries of $A$ now contain text positions $0,\dots, n/(c'+1)-1$ sorted by their fingerprint. We apply recursively the above procedure  to text positions $n/(c'+1),\dots, n-1$ using the free space left in $A$ (i.e. $n\log n - \frac{n}{c'+1}\log n$ bits) to perform sorting. 
	Note that, at each recursion step, the numbers we are sorting are always of fixed length ($(c'+1)\log n$ bits).
	We recurse on $n - \frac{n}{c'+1} = n\frac{c'}{c'+1}$ text positions. Let $d = \frac{c'+1}{c'} > 1$. This gives us the recurrence $T(n) = \bigO(n) + T(n / d)$ for our overall procedure (with base case $T(1) = \bigO(1)$), which results in overall $T(n) = \bigO(n)$ time (since $d>1$). 
	
	After terminating the above procedure, $A$ contains $\bigO(\log n)$ sub-arrays of text positions sorted by their fingerprints. Starting from the two rightmost such sub-arrays, we repeatedly apply in-place merge sort~\cite{SS1987} until the whole $A$ is sorted. Note that a single comparison of two text positions $i$ and $j$ requires computing $\rk(T[i,\dots, i+2^e-1])$ and $\rk(T[j,\dots, j+2^e-1])$ (constant time with our structure and using the pre-computed value $2^e \mod q$). Boundaries of the sub-arrays can be computed on-the-fly (i.e. $0, n - n/d, n - n/d^2 ,...$). Analogously to the above analysis, at the $j$-th step, $j\geq 0$, we merge in linear time two sub-arrays of total size $\bigO(d^j)$. The overall time spent inside this procedure is therefore $\bigO(\sum_{i=1}^{\log_dn}d^j) = \bigO(n)$. \qed 
\end{proof}

Note that we can always reduce the word size to $w \in \Theta(\log n)$. This, however, comes at the cost of limiting the alphabet size to $\sigma \leq n^{\bigO(1)}$ (as we assume alphabet characters fit in a constant number of memory words).
We can now use these results to build with a randomized algorithm a deterministic LCE data structure (i.e. that always returns the correct results). We randomly pick pairs $q,\seed$ 
and keep re-building our LCE structure until:
(1) its total space usage is of $n\lceil\log\sigma\rceil + \bigO(w)$ bits, and
(2) $\rk$ is collision-free over all pairs of substrings of $T$ having the same length $k=2^e$, for all $1\leq e\leq \log n$.
Checking property (1) can be done  during construction. As described in Section \ref{sec:in-place}, by reversing the construction whenever $S$ becomes non-empty, the working space never exceeds $n\lceil\log\sigma\rceil + \bigO(w)$ bits. After successful construction, property (2) can be checked with the space/time tradeoffs of Lemmas \ref{th:derand1}-\ref{th:derand3}. We are left to show what is the expected number $R$ of pairs $q,\seed$ we have to pick before both properties are satisfied. Using the results stated in Lemmas \ref{lemma: PDS} and \ref{lemma:wong LCE}, we prove the following:

\begin{lemma}\label{lemma:construction rounds}
	If we choose $\tau = 10w$, then we need to repeat the construction of our LCE structure $\bigO(1)$ expected times until (1) its total space usage is of $n\lceil\log\sigma\rceil + \bigO(w)$ bits, and
	(2) $\rk$ is collision-free over all pairs of substrings of $T$ having the same length $k=2^e$, for all $1\leq e\leq \log n$.
\end{lemma}
\begin{proof}
	Recall (proof of Lemma \ref{lemma:wong LCE}) that $C$ is the random variable denoting the number of pairs $\langle X,Y\rangle$ of $T$ substrings with $|X|=|Y|$ that generate a collision, i.e. $X\neq Y$ and $\rk(X)=\rk(Y)$, and that $S$ is the set containing all positions $i$ such that the most significant bit of $P'[i]$ is equal to 1. We are interested in computing a lower bound for the success probability
	\begin{equation}\label{eq:success prob}
	\mathcal P(C = 0 \wedge |S|= 0) = 1- \mathcal P(C>0 \vee |S|> 0)
	\end{equation}
	
	From the inequality $P(C> 0 \vee |S|> 0) \leq \mathcal P(C>0) + \mathcal P(|S|> 0)$, we obtain that the quantity in Equation \ref{eq:success prob} is greater than or equal to
	\begin{equation}\label{eq:success prob2}
	1 - \mathcal P(C>0) - \mathcal P(|S|> 0)
	\end{equation}

	We choose $\tau = 10w$. This satisfies Constraints (\ref{constraint2}),  (\ref{constraint1}), and (\ref{constraint3}) and implies that---see Section \ref{sec:det space}---$\mathcal P(|S|> 0) = 1 - \mathcal P(|S|=0) \leq 1 - 0.5 = 0.5$. It follows that quantity in Equation \ref{eq:success prob2} is greater than or equal to 
	$
	0.5 - \mathcal P(C>0)
	$. 
	Finally---see proof of Lemma \ref{lemma:wong LCE}---the choice $\tau = 10w$ implies $\mathcal P(C>0)\leq n^{-1}$. This, plugged into the above inequalities, gives us
	$
	\mathcal P(C = 0 \wedge |S|=0) \geq 0.5 - n^{-1}
	$. Note that $n^{-1} \leq 0.25$ holds for $n\geq 4$, which is true by our assumption $n\in\omega(1)$. We obtain: $\mathcal P(C = 0 \wedge |S|=0) \geq 0.25$.
	The number $R$ of rounds of our construction algorithm is a geometric random variable with success probability $p=\mathcal P(C = 0 \wedge |S|=0) \geq 0.25$, and has therefore expected value $1/p \leq 4$. \qed
\end{proof}

We use the technique exploited in Theorem \ref{th:MC structure 2}---i.e. we only compare substrings whose length is a power of 2 during exponential and binary search---in order to compute LCE queries with our structure, so that we only need $\rk$ to be collision-free between text substrings whose lengths are powers of two. If we do not pre-compute values $2^{\lceil\log\sigma\rceil\cdot 2^i}\mod q$, $0\leq i< \log n$, at each binary/exponential search step we have to compute one of them in $\bigO(\log \ell)$ time using the fast exponentiation algorithm. Using the collision-checking procedures described in Lemmas \ref{th:derand1}-\ref{th:derand3} we obtain:

\begin{theorem}\label{th:det structure 1}
	Within the following time-space bounds (space is on top of $T$):
	\begin{itemize}
		\item $\bigO(n\log n)$ expected time and $\bigO(n)$ words of space, or
		\item $\bigO(n\log^2 n)$ expected time and $n$ words of space, or
		\item $\bigO(n\log n)$ expected time and $n$ words of space --- provided that $w\in\Theta(\log n)$
	\end{itemize}
	we can replace the text with a deterministic data structure of $n\lceil \log\sigma\rceil + \bigO(w)$ bits supporting extraction of any length-$m$ text substring and LCE queries in $\bigO\left(m\log\sigma/w \right)$ and $\bigO(\log^2\ell)$ worst-case time, respectively. LCE queries are supported in $\bigO(\log\ell)$ time using $\bigO(\log n)$ additional words of space.  
\end{theorem}

Note that our data structure replaces the text and uses $s(n) = \Theta(\log n)$ or $s(n) = \Theta(\log^2 n)$ bits of space on top of the text (taking $w$ to be $\log n$) to support LCE queries in time $t(n) = \bigO(\log^2 n)$ and $t(n) = \bigO(\log n)$, respectively. Very recently, Kosolobov showed~\cite{kosolobov2017tight} that the relation $s(n)t(n) \in \Omega(n\log n)$ must hold when the text is read-only and $s(n) \in \Omega(n)$. Our results satisfy $s(n)t(n) \in \bigO(\log^3 n)$, but do not break the above lower bound in that (i) our space violates the requirement $s(n) \in \Omega(n)$, and (ii) our model allows the text to be overwritten.

\section{In-place LCP array, sparse suffix sorting, and suffix selection}\label{sec:in-place results}

In this section we attack the problems of computing in-place the LCP, the SSA, and the SLCP arrays. We assume that each entry of these arrays is stored using $\log n$ bits. We moreover describe the first in-place solution for the suffix selection problem: to return the $i$-th lexicographically smallest text suffix.

With \emph{slow LCE queries} on our LCE data structure we denote queries running in $\bigO(\log^2 n)$ time and $\bigO(w)$ bits of space on top of the text's space. With \emph{fast LCE queries} we denote those running in $\bigO(\log n)$ time and requiring $\tau\cdot\log n = 10w\log n$ bits of space on top of the text's space (this space is needed to store values  $2^{\lceil\log\sigma\rceil\cdot 2^i}\mod q$, see Theorem \ref{th:MC structure 2}). Note that the lexicographic order of any two text suffixes can be easily computed by comparing the two characters following their longest common prefix (i.e. one LCE query and one text access). 

In some of our results below, we obtain the space to support fast LCE queries by compressing integer sequences as done in~\cite{franceschini2007place,radixsort}: given a sequence $S[1,\dots,k]$ of $\log n$-bits integers, with $k\geq 10w\log n$, we first sort it in-place and in $\bigO(k\log k)$ time using any in-place comparison-sort algorithm. Then, we store in one word the index of the first integer in the sorted sequence starting with bit `1' and compact $S$ in $k\log n -k$ adjacent bits by removing the most significant bit from each integer. This saves $k \geq 10w\log n$ bits of space. We store values $2^{\lceil\log\sigma\rceil\cdot 2^i}\mod q$ in this space, so that our structure supports fast LCE queries. When fast LCE queries are no more needed, $S$ can be decompressed (i.e. pre-pending again the most significant bit to each integer). 

\subsection{In-place LCP array} 
In this paragraph we describe a sub-quadratic in-place LCP construction algorithm. The result, stated in the following theorem, follows from a careful combination of our deterministic LCE data structure, in-place suffix sorting~\cite{franceschini2007place}, in-place radix sorting~\cite{radixsort}, and compression of integer sequences:

\begin{theorem}\label{th_LCP}
	The Longest Common Prefix array ($LCP$) of a text $T\in\Sigma^n$  stored in $n\lceil\log\sigma\rceil$ bits can be computed within the following bounds:
	\begin{enumerate}
		\item $\bigO(n\log^2 n)$ expected time and $\bigO(1)$ words of space on top of the text and the $LCP$, provided that $|\Sigma| \leq 2^{\bigO(w)}$, or
		\item $\bigO(n\log n)$ expected time and $\bigO(1)$ words of space on top of the text and the $LCP$, provided that $|\Sigma|\leq n^{\bigO(1)}$.
	\end{enumerate}
\end{theorem}
\begin{proof}
	To get bound (1), we build and de-randomize our LCE structure using Lemma \ref{th:derand2}. We store in $\bigO(1)$ words the modulo $q$ and the seed $\seed$ computed during construction and restore the text. We build the suffix array SA in-place and $\bigO(n\log n)$ time using \cite{franceschini2007place}. We re-build our deterministic LCE structure using the values $q$ and $\seed$ computed above (in-place and $\bigO(n)$ time). We convert SA to LCP using slow LCE queries to compute LCE's of adjacent suffixes (in-place and $\bigO(n\log^2 n)$ time). To conclude, we restore the text inverting the construction of our structure. 
	To get bound (2), we must be more careful. First of all, we limit the word size to $w=\log n$. Then:
	\begin{itemize}
		\item We build and de-randomize our LCE structure using Lemma \ref{th:derand3}. We store in $\bigO(1)$ words the modulo $q$ and the seed $\seed$ computed during construction and restore the text (needed for the next step).
		\item We build the suffix array SA in-place and $\bigO(n\log n)$ time using \cite{franceschini2007place}. Note that we cannot perform this step before computing  $q$ and $\seed$, as the de-randomization procedures needs $n$ words of space. 
		\item We re-build our deterministic LCE structure using the values $q$ and $\seed$ computed in the first step. This step runs in-place and $\bigO(n)$ time.
		\item We compress the integers in $SA' = SA[1,\dots, n/\log^2 n]$ with the procedure described at the beginning of this section (in-place and $\bigO(n/\log n)$ time).  This saves $n/\log^2 n > 10\log^2n$ bits of space (this inequality holds for $n$ larger than some constant). We store values $2^{\lceil\log\sigma\rceil\cdot 2^i}\mod q$ in this space, so that our structure now supports fast LCE queries. Note that $SA'$ is no more suffix-sorted.
		\item We convert $SA''=SA[n/\log^2 n+1,\dots, n]$ to  $LCP[n/\log^2 n+1,\dots, n]$ by computing LCE values of adjacent suffixes using fast LCE queries. This step runs in-place and $\bigO(n\log n)$ time.
		\item We decompress $SA'$. Now our LCE structure supports only slow LCE queries.
		\item We suffix-sort in-place (comparison-based sorting) $SA'$ using slow LCE queries. This step takes $\bigO((n/\log^2n)\log(n/\log^2n)\log^2 n) = \bigO(n\log n)$ time.
		\item We convert $SA'$ to $LCP[1, \dots, n/\log^2n]$ by computing LCE values of adjacent suffixes using slow LCE queries. This step runs in $\bigO\left(\frac{n}{\log^2n}\cdot \log^2 n\right) = \bigO(n)$ time and in-place.
	\end{itemize}
\end{proof}

\subsection{In-place SSA and SLCP arrays}

By combining our Monte Carlo data structure with in-place comparison sorting, in-place comparison-based merging~\cite{SS1987}, and compression of integer sequences we  obtain:

\begin{theorem}\label{th:suffix-sort}
	Any set $S=\{i_1,\dots,i_b\}$ of $b$ suffixes of a text $T\in\Sigma^n$ stored in $n\lceil\log\sigma\rceil$ bits, with $|\Sigma|\leq n^{\bigO(1)}$, can be sorted correctly with high probability in  $\bigO(n+b\log^2 n)$ expected time using $\bigO(1)$ words of space on top of $T$ and $S$.
\end{theorem}
\begin{proof} 
	We limit the word size to $w=\log n$, with the resulting requirement $|\Sigma|\leq n^{\bigO(1)}$.
	We first build  our Monte Carlo structure using the in-place construction algorithm of Section \ref{sec:in-place} ($\bigO(n)$ expected time).
	If $b<n/\log^3 n$, then we suffix-sort the $b$ text positions plugging slow LCE queries in any in-place comparison sorting algorithm terminating within $\bigO(b\log b)$ comparisons. This requires $\bigO(b\log b\log^2 n)\in\bigO(n)$ time. 
	If $b\geq n/\log^3 n$, then we divide  $S$  in two sub-arrays $S' = S[1,\dots,n/\log^3 n]$ and $S'' = S[n/\log^3 n+1,\dots, b]$ and: 
	\begin{itemize}
		\item We compress $S'$ in-place and $\bigO(n/\log^2n)$ time as described at the beginning of this section. This saves $n/\log^3 n > 10\log^2n$ bits of space (note that this inequality holds for $n$ larger than some constant). We store values $2^{\lceil\log\sigma\rceil\cdot 2^i}\mod q$ in this space, so that our structure now supports fast LCE queries.
		\item We suffix-sort in-place (comparison-based sorting) $S''$ using fast LCE queries. This step terminates in $\bigO(b\log b\log n) \subseteq \bigO(b\log^2n)$ time. 
		\item We decompress $S'$. Now our structure supports only slow LCE queries.
		\item We suffix sort in-place (comparison-based sorting) $S'$ using slow LCE queries. This step terminates in $\bigO((n/\log^3 n)\cdot \log(n/\log^3 n)\cdot \log^2 n)\in\bigO(n)$ time
		\item We merge $S'$ and $S''$ using slow LCE queries and in-place comparison-based merging~\cite{SS1987}. This step terminates in $\bigO(b\log^2 n)$ time.
	\end{itemize}
\end{proof}

Note that we can add a further step to the above-described procedure: after obtaining the SSA, we can overwrite it with the SLCP by replacing adjacent SSA entries with their LCE. This step runs in-place and $\bigO(b\log^2n)$ time using slow LCE queries. We obtain:

\begin{theorem}\label{th:sparse LCP}
	Given a text $T\in\Sigma^n$  stored in $n\lceil\log\sigma\rceil$ bits, with $|\Sigma|\leq n^{\bigO(1)}$,  and a subset $S=\{i_1,\dots,i_b\}$ of $b$ suffixes of $T$, we can---with high probability of success---replace $S$ with the sparse LCP array relative to $S$ in  $\bigO(n+b\log^2 n)$ expected time using $\bigO(1)$ words of space on top of $T$ and $S$.
\end{theorem}

\subsection{In-place suffix selection}

In this paragraph we provide the first optimal-space sub-quadratic algorithm solving the suffix selection problem: given a text $T$ and an index $0\leq i <n$, output the text position corresponding to the $i$-th lexicographically smallest text suffix. Our idea is to use a variant of the quick-select algorithm (i.e. solving the selection problem on integers) operating with the lexicographical ordering of text suffixes. 

\begin{theorem}
	Given a text $T\in\Sigma^n$  stored in $n\lceil\log\sigma\rceil$ bits, with $|\Sigma|\leq 2^{\bigO(w)}$, and an index $0\leq i <n$, we can---with high probability of success---find the $i$-th lexicographically smallest text suffix in $\bigO(n\log^3n)$ expected time using $\bigO(1)$ words of working space on top of $T$.
\end{theorem}
\begin{proof}
	We build  our Monte Carlo structure using the in-place construction algorithm of Section \ref{sec:in-place} ($\bigO(n)$ expected time). In our procedure below, we use slow LCE queries to lexicographically compare pairs of suffixes.
	We scan $T$ and find the lexicographically smallest and largest text suffixes $i_{min}$, $i_{max}$ in $\bigO(n\log^2n)$ time\footnote{i.e. for every $0\leq j < n$ we compare the $j$-th suffix with the $i_{min}$-th and $i_{max}$-th suffixes and determine whether $j$ is the new $i_{min}$ or $i_{max}$. At the beginning, we can start with $i_{min} = min_{lex}(0,1)$ and $i_{max} = max_{lex}(0,1)$.}. Then, we scan again $T$ and count the number $m$ of suffixes inside the lexicographic range $[i_{min}, i_{max}]$ (at the beginning, $m=n$). We pick a uniform random number $r$ in $[1,m]$, scan again the text, and select the $r$-th suffix $i_r$ falling inside lexicographic range $[i_{min}, i_{max}]$ that we see during the scan. Finally, we recurse on $[i_{min},i_r]$ or $[i_r,i_{max}]$ depending on where the $i$-th smallest suffix falls (as in quick select, we repeat the process of choosing $r$ until we recurse on a set with at most $3m/4$ elements. Since $i_r$ is uniform inside range $[i_{min}, i_{max}]$, we need to pick at most $\bigO(1)$ values of $r$ until this condition is satisfied). We perform at most $\bigO(\log n)$ recursive steps, therefore the expected running time of our algorithm is $\bigO(n\log^3n)$.	
\end{proof}

\section{Conclusions and open problems}

In this paper we have described the first in-place solutions to the following problems: (i) sparse suffix-sorting, (ii) sparse LCP array construction, and (iii) suffix selection. To achieve our results, we have introduced a novel \emph{implicit} LCE data structure whose space is only a constant number of memory words higher than that of the text. The text can be replaced in-place and linear time with our data structure; this transformation is a general and powerful tool that --- in particular --- can be used to solve in-place problems (i-iii). 

Interestingly, previous strategies solving in-place the suffix sorting problem~\cite{franceschini2007place,goto2017optimal,li2016optimal} re-use only the space of the suffix array (through integer compression) to achieve constant working space. In this work, we re-use \emph{both} the space of the text (replacing it with a more powerful representation) and of the sparse suffix array (through integer compression) to achieve the more ambitious goal of sorting in-place any subset of text positions. It is therefore natural to ask whether or not it is possible to achieve the same goal without overwriting the text.

Note that only our algorithms for computing in-place the \emph{full} LCP and suffix arrays return always the correct result. Our other solutions return the correct answer with high probability only.
One  important improvement over our work could therefore be that of obtaining Las Vegas versions of our algorithms. Crucially, this should be achieved without increasing working space (since asymptotically optimal solutions are already known). One possible line of attack  could be to devise an in-place procedure for checking the correctness of a sparse suffix array. Alternatively, one could try to verify that the Karp-Rabin hash function used in our LCE data structure is collision-free on text substrings relevant to the suffix-sorting procedure. To this extent, we note that the $b$ input text positions could be used to encode extra information needed in the de-randomization process; we already showed that this is possible if $b=n$. The case $n - b = \omega(1)$ remains an intriguing open problem.

\subsection{Acknowledgments}

I wish to thank Tomasz Kociumaka and anonymous reviewers for useful comments on a preliminary version of the paper.


\bibliographystyle{plain}
\bibliography{40}

\begin{thebibliography}{10}

\bibitem{baeza1992new}
Ricardo Baeza-Yates and Gaston~H Gonnet.
\newblock A new approach to text searching.
\newblock {\em Communications of the ACM}, 35(10):74--82, 1992.

\bibitem{bille2013sparse}
Philip Bille, Johannes Fischer, Inge~Li G{\o}rtz, Tsvi Kopelowitz, Benjamin
  Sach, and Hjalte~Wedel Vildh{\o}j.
\newblock Sparse text indexing in small space.
\newblock {\em ACM Transactions on Algorithms (TALG)}, 12(3):39, 2016.

\bibitem{bille2015longest}
Philip Bille, Inge~Li G{{\o}}rtz, Mathias B{{\ae}}k~Tejs Knudsen, Moshe
  Lewenstein, and Hjalte~Wedel Vildh{{\o}}j.
\newblock {Longest common extensions in sublinear space}.
\newblock In {\em {Annual Symposium on Combinatorial Pattern Matching}}, pages
  {65--76}. {Springer}, {2015}.

\bibitem{bille2014time}
Philip Bille, Inge~Li G{\o}rtz, Benjamin Sach, and Hjalte~Wedel Vildh{\o}j.
\newblock {Time--space trade-offs for longest common extensions}.
\newblock {\em {Journal of Discrete Algorithms}}, 25:42--50, 2014.

\bibitem{dodis2010changing}
Yevgeniy Dodis, Mihai Patrascu, and Mikkel Thorup.
\newblock Changing base without losing space.
\newblock In {\em Proceedings of the forty-second ACM symposium on Theory of
  computing}, pages 593--602. ACM, 2010.

\bibitem{fischer2016deterministic}
Johannes Fischer, I~Tomohiro, and Dominik K{\"o}ppl.
\newblock {Deterministic sparse suffix sorting on rewritable texts}.
\newblock In {\em Latin American Symposium on Theoretical Informatics}, pages
  483--496. Springer, 2016.

\bibitem{fouque2014close}
Pierre-Alain Fouque and Mehdi Tibouchi.
\newblock Close to uniform prime number generation with fewer random bits.
\newblock In {\em International Colloquium on Automata, Languages, and
  Programming}, pages 991--1002. Springer, 2014.

\bibitem{franceschini2007place}
Gianni Franceschini and S~Muthukrishnan.
\newblock In-place suffix sorting.
\newblock In {\em International Colloquium on Automata, Languages, and
  Programming}, pages 533--545. Springer, 2007.

\bibitem{franceschini2007optimal}
Gianni Franceschini and S~Muthukrishnan.
\newblock Optimal suffix selection.
\newblock In {\em Proceedings of the thirty-ninth annual ACM symposium on
  Theory of computing}, pages 328--337. ACM, 2007.

\bibitem{radixsort}
Gianni Franceschini, S.~Muthukrishnan, and Mihai Patrascu.
\newblock Radix sorting with no extra space.
\newblock In {\em Proc. 15th ESA}, pages 194--205, 2007.

\bibitem{gawrychowski2017sparse}
Pawe{\l} Gawrychowski and Tomasz Kociumaka.
\newblock Sparse suffix tree construction in optimal time and space.
\newblock In {\em Proceedings of the Twenty-Eighth Annual ACM-SIAM Symposium on
  Discrete Algorithms}, pages 425--439. SIAM, 2017.

\bibitem{gog2011fast}
Simon Gog and Enno Ohlebusch.
\newblock {Fast and lightweight LCP-array construction algorithms}.
\newblock In {\em Proceedings of the Meeting on Algorithm Engineering \&
  Expermiments}, pages 25--34. Society for Industrial and Applied Mathematics,
  2011.

\bibitem{gonnet1992new}
Gaston~H Gonnet, Ricardo~A Baeza-Yates, and Tim Snider.
\newblock New indices for text: Pat trees and pat arrays.
\newblock {\em Information Retrieval: Data Structures \& Algorithms}, 66:82,
  1992.

\bibitem{goodrich2014algorithm}
Michael~T Goodrich and Roberto Tamassia.
\newblock {\em Algorithm design and applications, Section 6.3.3: Linear
  Probing}.
\newblock Wiley Publishing, 2014.

\bibitem{goto2017optimal}
Keisuke Goto.
\newblock Optimal time and space construction of suffix arrays and lcp arrays
  for integer alphabets.
\newblock {\em arXiv preprint arXiv:1703.01009}, 2017.

\bibitem{heath1978differences}
DR~Heath-Brown.
\newblock The differences between consecutive primes.
\newblock {\em Journal of the London Mathematical Society}, 2(1):7--13, 1978.

\bibitem{karkkainen2014faster}
Tomohiro I, Juha K\"arkk\"ainen, and Dominik Kempa.
\newblock {Faster Sparse Suffix Sorting}.
\newblock In {\em 31st International Symposium on Theoretical Aspects of
  Computer Science}, pages 386--396. Schloss Dagstuhl-Leibniz-Zentrum fuer
  Informatik, 2014.

\bibitem{inenaga2015faster}
Shunsuke Inenaga.
\newblock {A faster longest common extension algorithm on compressed strings
  and its applications}.
\newblock In {\em {Prague Stringology Conference}}, pages {1--4}, {2015}.

\bibitem{karkkainen2006linear}
Juha K{\"a}rkk{\"a}inen, Peter Sanders, and Stefan Burkhardt.
\newblock Linear work suffix array construction.
\newblock {\em Journal of the ACM (JACM)}, 53(6):918--936, 2006.

\bibitem{karkkainen1996sparse}
Juha K{\"a}rkk{\"a}inen and Esko Ukkonen.
\newblock Sparse suffix trees.
\newblock {\em Computing and Combinatorics}, pages 219--230, 1996.

\bibitem{karp1987efficient}
Richard~M Karp and Michael~O Rabin.
\newblock {Efficient randomized pattern-matching algorithms}.
\newblock {\em {IBM Journal of Research and Development}},
  {31}({2}):{249--260}, {1987}.

\bibitem{knuthbook}
Donald Knuth.
\newblock {\em The Art of Computer Programming, Seminumerical Algorithms},
  volume~2.
\newblock Addison-Wesley, 1981.

\bibitem{kosolobov2017tight}
Dmitry Kosolobov.
\newblock Tight lower bounds for the longest common extension problem.
\newblock {\em Information Processing Letters}, 125:26--29, 2017.

\bibitem{li2016optimal}
Zhize Li, Jian Li, and Hongwei Huo.
\newblock Optimal in-place suffix sorting.
\newblock {\em arXiv preprint arXiv:1610.08305}, 2016.

\bibitem{LOUZA201714}
Felipe~A. Louza, Travis Gagie, and Guilherme~P. Telles.
\newblock {Burrows-Wheeler transform and LCP array construction in constant
  space}.
\newblock {\em Journal of Discrete Algorithms}, 42(Supplement C):14 -- 22,
  2017.

\bibitem{maier1985primes}
Helmut Maier et~al.
\newblock Primes in short intervals.
\newblock {\em The Michigan Mathematical Journal}, 32(2):221--225, 1985.

\bibitem{manber1993suffix}
Udi Manber and Gene Myers.
\newblock Suffix arrays: a new method for on-line string searches.
\newblock {\em siam Journal on Computing}, 22(5):935--948, 1993.

\bibitem{nishimoto2015dynamic}
Takaaki Nishimoto, I~Tomohiro, Shunsuke Inenaga, Hideo Bannai, and Masayuki
  Takeda.
\newblock {Fully Dynamic Data Structure for LCE Queries in Compressed Space}.
\newblock In {\em {41st International Symposium on Mathematical Foundations of
  Computer Science}}, pages {72:1--72:15}, {2015}.

\bibitem{puglisi2007taxonomy}
Simon~J Puglisi, William~F Smyth, and Andrew~H Turpin.
\newblock A taxonomy of suffix array construction algorithms.
\newblock {\em acm Computing Surveys (CSUR)}, 39(2):4, 2007.

\bibitem{puglisi2008space}
Simon~J Puglisi and Andrew Turpin.
\newblock Space-time tradeoffs for longest-common-prefix array computation.
\newblock In {\em International Symposium on Algorithms and Computation}, pages
  124--135. Springer, 2008.

\bibitem{SS1987}
Jeffrey~S. Salowe and William~L. Steiger.
\newblock Simplified stable merging tasks.
\newblock {\em J. Algorithms}, 8(4):557--571, 1987.

\bibitem{tanimura2016}
Yuka Tanimura, Tomohiro I, Hideo Bannai, Shunsuke Inenaga, Simon~J. Puglisi,
  and Masayuki Takeda.
\newblock {Deterministic sub-linear space LCE data structures with efficient
  construction}.
\newblock In {\em {Annual Symposium on Combinatorial Pattern Matching}}.
  {Springer}, {2016}.

\bibitem{tanimura2017small}
Yuka Tanimura, Takaaki Nishimoto, Hideo Bannai, Shunsuke Inenaga, and Masayuki
  Takeda.
\newblock Small-space encoding lce data structure with constant-time queries.
\newblock {\em arXiv preprint arXiv:1702.07458}, 2017.

\bibitem{ukkonen1995line}
Esko Ukkonen.
\newblock On-line construction of suffix trees.
\newblock {\em Algorithmica}, 14(3):249--260, 1995.

\end{thebibliography}

\end{document}